\documentclass[12pt]{amsart}
\usepackage{amssymb}
\usepackage{color}
\usepackage{amsmath,epic,curves,amscd}
\usepackage[english]{babel}
\usepackage{graphicx}
\usepackage{comment}
\usepackage{appendix}
\usepackage{mathdots}
\usepackage[all]{xy}
\newtheorem{claim}{}[section]
\newtheorem{theorem}[claim]{Theorem}

\newtheorem{lemma}[claim]{Lemma}
\newtheorem{proposition}[claim]{Proposition}
\newtheorem{corollary}[claim]{Corollary}

\theoremstyle{remark}

\renewenvironment{proof}{\noindent{\it Proof. \hskip0pt}}
                      {$\square$\par\medskip}

\begin{document}
\baselineskip 6.0 truemm
\parindent 1.5 true pc

\newcommand\lan{\langle}
\newcommand\ran{\rangle}
\newcommand\tr{{\text{\rm Tr}}\,}
\newcommand\ot{\otimes}
\newcommand\ol{\overline}
\newcommand\join{\vee}
\newcommand\meet{\wedge}
\renewcommand\ker{{\text{\rm Ker}}\,}
\newcommand\image{{\text{\rm Im}}\,}
\newcommand\id{{\text{\rm id}}}
\newcommand\tp{{\text{\rm tp}}}
\newcommand\pr{\prime}
\newcommand\e{\epsilon}
\newcommand\la{\lambda}
\newcommand\inte{{\text{\rm int}}\,}
\newcommand\ttt{{\text{\rm t}}}
\newcommand\spa{{\text{\rm span}}\,}
\newcommand\conv{{\text{\rm conv}}\,}
\newcommand\rank{\ {\text{\rm rank of}}\ }
\newcommand\re{{\text{\rm Re}}\,}
\newcommand\ppt{\mathbb T}
\newcommand\rk{{\text{\rm rank}}\,}
\newcommand\SN{{\text{\rm SN}}\,}
\newcommand\SR{{\text{\rm SR}}\,}
\newcommand\HA{{\mathcal H}_A}
\newcommand\HB{{\mathcal H}_B}
\newcommand\HC{{\mathcal H}_C}
\newcommand\CI{{\mathcal I}}
\newcommand{\bra}[1]{\langle{#1}|}
\newcommand{\ket}[1]{|{#1}\rangle}
\newcommand\cl{\mathcal}
\newcommand\idd{{\text{\rm id}}}
\newcommand\OMAX{{\text{\rm OMAX}}}
\newcommand\OMIN{{\text{\rm OMIN}}}
\newcommand\diag{{\text{\rm Diag}}\,}
\newcommand\calI{{\mathcal I}}
\newcommand\bfi{{\bf i}}
\newcommand\bfj{{\bf j}}
\newcommand\bfk{{\bf k}}
\newcommand\bfl{{\bf l}}
\newcommand\bfzero{{\bf 0}}
\newcommand\bfone{{\bf 1}}

\title{Construction of multi-qubit optimal genuine entanglement witnesses}

\author{Kyung Hoon Han and Seung-Hyeok Kye}
\address{Department of Mathematics, The University of Suwon, Gyeonggi-do 445-743, Korea}
\email{kyunghoon.han at gmail.com}
\address{Department of Mathematics and Institute of Mathematics, Seoul National University, Seoul 151-742, Korea}
\email{kye at snu.ac.kr}
\thanks{KHH and SHK were partially supported by NRF-2012R1A1A1012190 and NRFK 2009-0083521., respectively.}



\begin{abstract}
We interpret multi-partite genuine entanglement witnesses as simultaneous positivity
of various maps arising from them. We apply this result to multi-qubit {\sf X}-shaped
Hermitian matrices, and characterize the conditions
for them to be genuine entanglement witnesses, in terms of entries. Furthermore, we find
all optimal ones among them. They turn out to have the spanning
properties, and so they detect non-zero volume set of multi-qubit
genuine entanglement. We also characterize decomposability for {\sf
X}-shaped entanglement witnesses.
\end{abstract}

\maketitle

\section{Introduction}

The notion of entanglement is considered as the main resource in current quantum information theory,
and it is one of the key research topics to study how to detect entanglement from separability.
The duality between separable states and positive linear maps turned out to be very useful
for this purpose \cite{eom-kye,horo-1}, and was formulated as the notion of entanglement witnesses \cite{terhal}
in the bi-partite case.
Entanglement witnesses also have the obvious natural meaning in the multi-partite cases,
and they are now interpreted as the Choi matrices of positive multi-linear maps by the
duality between $n$-partite separable states and positive multi-linear maps with $(n-1)$ variables \cite{kye_multi_dual}.
In the multi-partite cases, there are many kinds of entanglement \cite{abls}.
Among them, genuine entanglement, especially multi-qubit genuine entanglement,
is at the central interest for quantum information processing,
and many authors suggested various methods to detect genuine entanglement from bi-separability.
See survey articles \cite{GT-rev,HHHH-rev}.
The notion of entanglement witnesses is also naturally extended in this case, and has been considered
by several authors. See \cite{{bancal},{berg13},{chenchen},{guhne11},{huber14},{jung},{kor05},{ryu},{toth05},{wu}},
for examples.

A multi-partite state $\varrho$ in the tensor product $\bigotimes_{i=1}^n M_{d_i}$ of matrix algebras $M_{d_i}$
on the $d_i$-dimensional Hilbert space $\mathbb C^{d_i}$ is
said to be {\sl (fully) separable} if it can be written as the convex combination
$$
\varrho
=\sum_k p_k |z_k\ran\lan z_k|
$$
of pure states $|z_k\ran\lan z_k|$ onto product vectors $|z_k\ran$, that is, simple tensors in
the tensor product $\bigotimes_{i=1}^n \mathbb C^{d_i}$ of Hilbert spaces. For a given bi-partition
$S\sqcup T$ of the set $[n]:=\{1,2,\dots,n\}$, a multi-partite state $\varrho$ may be considered as a bi-partite state
in the tensor product
$\left(\bigotimes_{i\in S} M_{d_i}\right)\otimes
\left(\bigotimes_{i\in T} M_{d_i}\right)$
of two matrix algebras,
and is said to be {\sl $S$-$T$ bi-separable} (respectively {\sl $S$-$T$ PPT})
if this bi-partite state is separable (respectively PPT). A multi-partite state $\varrho$ is called {\sl bi-separable}
(respectively {\sl a PPT mixture}) if
it is in the convex hull of $S$-$T$ bi-separable (respectively $S$-$T$ PPT) states through all bi-partitions $S\sqcup T=[n]$.
A state is said to be {\sl genuinely entangled} if it is not bi-separable.
We call a non-positive (non positive semi-definite) Hermitian matrix $W$ in $\bigotimes_{i=1}^n M_{d_i}$
{\sl genuine entanglement witness} if
$$
\lan \varrho, W\ran:=\tr (\varrho W^\ttt)\ge 0
$$
for every bi-separable state $\varrho$. Here, $W^\ttt$ denotes the transpose of $W$.
Non-positivity condition of $W$ guarantees the existence of a state $\varrho$ with $\lan\varrho,W\ran <0$, and the above condition
tells us that this $\varrho$ must be genuinely entangled.
By duality, any genuine entanglement is detected by a genuine entanglement witness.

In the three partite case,
the authors \cite{han_kye_tri} interpreted genuine entanglement witnesses as the Choi matrices
of $(p,q,r)$-positive bi-linear maps, and constructed various kinds of three-qubit entanglement witnesses.
In this paper, we interpret general multi-partite genuine entanglement witnesses
in terms of simultaneous positivity of various linear maps arising from bi-partitions $S\sqcup T=[n]$.
For this purpose, it is very important to set up notations. We will do it in the next section,
and describe the linear maps arising from bi-partitions. We also summarize results in this paper
in terms of these notations.

Motivated by examples of three-qubit entanglement witnesses constructed in \cite{{han_kye_tri}}{ and \cite{{kye_multi_dual}},
we apply this result in Section 3 to the so called {\sf X}-shaped multi-qubit witnesses
to characterize genuine entanglement witnesses in terms of entries.
Recall that a matrix is {\sf X}-{\sl shaped} if all the entries are
zero except for diagonal and anti-diagonal entries. Among genuine entanglement
witnesses we found, we characterize in Section 4 optimal ones which
turn out to have the spanning properties. Therefore, they detect
nontrivial set of genuine entanglement, that is, the set of genuine
entanglement detected by them have non-zero volume. In Section 5, we
also characterize decomposable witnesses, and see that every {\sf X}-shaped genuine entanglement witness
is decomposable.

We note that states with {\sf X}-shaped matrix forms have been studied by several authors
in various contexts. See \cite{{abls},{guhne10},{mendo},{Rafsanjani},{rau},{vin10},{wein10},{yu},{yu07}} for example.

\section{Notations and summary of results}

In order to deal with multi-partite systems, it is convenient to use multi-indices for entries of matrices.
Let $S$ be a nonempty subset of $[n]=\{1,2,\dots,n\}$.
A function $\bfi $ from $S$ into nonnegative integers with $0\le {\bfi}(i)<
d_i$ ($i\in S$) will be called an {\sl index} on $S$, which will be
denoted by a string of integers in the obvious sense.

For a given bi-partition $S\sqcup T=[n]$, it is clear with this notation
that any matrix $W$ in $\bigotimes_{i\in [n]} M_{d_i}$
can be written in a unique way by
$$
W=\sum_{{\bfi},{\bfj}\in I_S} |{\bfi}\ran\lan {\bfj}| \ot W[{\bfi},{\bfj}]\in
\left(\bigotimes_{i\in S} M_{d_i}\right)\otimes \left(\bigotimes_{i\in T} M_{d_i}\right)
$$
where $I_S$ denotes the set of all indices on the set $S$. We also use the notation
$$
|{\bfi}\ran=|i_1\ran\ot |i_2\ran\ot \cdots\ot |i_{\# S}\ran
$$
for ${\bfi}=i_1i_2\dots i_{\# S}\in I_S$, and similarly for $\lan {\bfj}|$, where $\#S$ denotes the cardinality of $S$.
For given indices ${\bfk},{\bfl}$ on $T$,
the $({\bfk},{\bfl})$-entry $W[{\bfi},{\bfj}]_{{\bfk},{\bfl}}$ of $W[{\bfi},{\bfj}]\in \bigotimes_{i\in T} M_{d_i}$ is given by
$$
W[{\bfi},{\bfj}]_{{\bfk},{\bfl}}=W_{{\bfi} \diamond {\bfk}, {\bfj} \diamond {\bfl}},
$$
where ${\bfi}\diamond {\bfk}$ is the index on $[n]$ defined by
$$
({\bfi} \diamond {\bfk})(i )=\begin{cases} {\bfi}(i ),& i \in S,\\ {\bfk}(i ),&i \in T.\end{cases}
$$

We note that the set $\{|{\bfi}\ran\lan {\bfj}| : {\bfi}, {\bfj} \in I_S \}$ plays the role of matrix units
for the matrix algebra $\bigotimes_{i\in S} M_{d_i}$. Therefore, we may define the linear map
\begin{equation}\label{st_map}
\phi_W^{S,T}: |{\bfi}\ran\lan {\bfj}| \in \bigotimes_{i\in S} M_{d_i}\to W[{\bfi},{\bfj}]
\in \bigotimes_{i\in T} M_{d_i}, \qquad {\bfi}, {\bfj} \in I_S.
\end{equation}
Conversely, for any given linear map $\phi : \bigotimes_{i\in S} M_{d_i} \to \bigotimes_{i\in T} M_{d_i}$,
we can associate the matrix $W_\phi\in\bigotimes_{i\in[n]}M_{d_i}$ by
$$
W_\phi=
\sum_{{\bfi},{\bfj} \in I_S} |{\bfi}\ran\lan {\bfj}| \ot \phi(|{\bfi}\ran\lan {\bfj}|)\in
\left(\bigotimes_{i\in S} M_{d_i}\right)\otimes \left(\bigotimes_{i\in T} M_{d_i}\right)=\bigotimes_{i\in[n]}M_{d_i}.
$$
When $n=2$ and $S=\{1\}$, $W_\phi$ is nothing but the usual Choi matrix \cite{choi75-10} of the linear map $\phi$ from $M_{d_1}$
into $M_{d_2}$.

We consider an example.
In the three qubit case, every {\sf X}-shaped matrix $W \in M_{d_1}\ot M_{d_2}\ot M_{d_3}$ with $d_i=2$
can be written by the usual $8\times 8$ matrix
$$
W=
\left(
\begin{matrix}
W_{000,000} &\cdot &\cdot &\cdot &\cdot &\cdot &\cdot &W_{000,111} \\
\cdot &W_{001,001} &\cdot &\cdot &\cdot &\cdot &W_{001,110} &\cdot \\
\cdot &\cdot &W_{010,010} &\cdot &\cdot &W_{010,101} &\cdot &\cdot \\
\cdot &\cdot &\cdot &W_{011,011} &W_{011,100} &\cdot &\cdot &\cdot \\
\cdot &\cdot &\cdot &W_{100,011} &W_{100,100} &\cdot &\cdot &\cdot \\
\cdot &\cdot &W_{101,010} &\cdot &\cdot &W_{101,101} &\cdot &\cdot \\
\cdot &W_{110,001} &\cdot &\cdot &\cdot &\cdot &W_{110,110}  &\cdot \\
W_{111,000} &\cdot &\cdot &\cdot &\cdot &\cdot &\cdot &W_{111,111}
\end{matrix}
\right),
$$
if we endow indices with the lexicographic order,
where $\cdot$ denotes zero.
The map $\phi_W^{\{2\},\{1,3\}}$ is a linear map from $M_{d_2}$ into $M_{d_1}\ot M_{d_3}$, and
the image of $|i\ran\lan j|\in M_{d_2}=M_2$ can be obtained by searching for the entries
which look like $W_{{\sf *}i{\sf *},{\sf *}j{\sf *}}$.
For example, the image of 
$|0\ran\lan 1|$ under $\phi_W^{\{2\},\{1,3\}}$ is given by
$$
\left(
\begin{matrix}
\cdot &\cdot &\cdot &W_{0{\bfzero}0,1{\bfone}1} \\
\cdot &\cdot &W_{0{\bfzero}1,1{\bfone}0} &\cdot \\
\cdot &W_{1{\bfzero}0,0{\bfone}1} &\cdot &\cdot \\
W_{1{\bfzero}1,0{\bfone}0} &\cdot &\cdot &\cdot \\
\end{matrix}
\right)\in M_{d_1}\otimes M_{d_3}=M_4.
$$

For a linear map $\phi : \bigotimes_{i\in S} M_{d_i} \to \bigotimes_{i\in T} M_{d_i}$
and a bi-partite state $\varrho\in \left(\bigotimes_{i\in S} M_{d_i}\right)\otimes \left(\bigotimes_{i\in T} M_{d_i}\right)$,
we have the the bilinear pairing $\lan\varrho,\phi\ran$
which coincides with
$\lan\varrho, W_\phi\ran := \tr (\varrho W^\ttt)$.
By the duality between bi-partite separability and positivity of linear maps, it is now clear that
the following are equivalent for a given bi-partition $S\sqcup T=[n]$:
\begin{itemize}
\item
$\lan\varrho, W\ran\ge 0$ for every $S$-$T$ bi-separable state $\varrho$.
\item
The linear map $\phi_W^{S,T}$ is positive.
\end{itemize}
We employ the duality between the convex hulls and the intersections, to get the
equivalence (i) $\Longleftrightarrow$ (ii) in the following:

\begin{proposition}
Let $W$ be a Hermitian matrix in $\bigotimes_{i\in [n]}M_{d_i}$. Then the following are equivalent:
\begin{enumerate}
\item[(i)]
$\lan\varrho, W\ran\ge 0$ for every bi-separable state $\varrho$.
\item[(ii)]
The linear map $\phi_W^{S,T}$ is positive for each bi-partition $S\sqcup T$ of $[n]$.
\item[(iii)]
The linear map $\phi_W^{S,T}$ is positive for each bi-partition $S\sqcup T$ of $[n]$ with $\# S\le \frac n2$.
\end{enumerate}
\end{proposition}

The statement (iii) is equivalent to (ii), because the map $\phi_W^{T,S}$ is the transpose of
$\phi_W^{S,T}$ whenever $S\sqcup T$ is a bi-partition of $[n]$. Indeed,
for given ${\bfi},{\bfj} \in I_S$ and ${\bfk},{\bfl}\in I_T$, we have
$$
\lan \phi_W^{S,T}(|{\bfi}\ran\lan {\bfj}|), |{\bfk}\ran\lan {\bfl}|\ran
=\lan W[{\bfi},{\bfj}],  |{\bfk}\ran\lan {\bfl}|\ran=W_{{\bfi} \diamond {\bfk}, {\bfj} \diamond {\bfl}},
$$
where $\lan\ \cdot\ , \ \cdot\ \ran$ is the bilinear pairing on the matrix algebra $\bigotimes_{i\in T} M_{d_i}$.
On the other hand, we also have
$$
\lan |{\bfi}\ran\lan {\bfj}|, \phi_W^{T,S}(|{\bfk}\ran\lan {\bfl}|)\ran
=\lan |{\bfi}\ran\lan {\bfj}|, W[{\bfk},{\bfl}])\ran
=W_{{\bfi} \diamond {\bfk}, {\bfj} \diamond {\bfl}},
$$
where $\lan\ \cdot\ , \ \cdot\ \ran$ is the bilinear pairing on the matrix algebra $\bigotimes_{i\in S} M_{d_i}$.

A multi-qubit matrix $X=[X_{\bfi,\bfj}]\in \bigotimes_{i=1}^n M_{d_i}$ with indices $\bfi,\bfj$ on $[n]$ is {\sf X}-shaped if and only if
$X_{\bfi,\bfj}$ is nonzero only when $\bfi=\bfj$ or $\bfi=\bar\bfj$, where
$$
\bar{\bfi}(i)=i+1\mod 2,\qquad i =1,2,\dots, n.
$$
For each  index ${\bfi}$ beginning with $0$ and numbers
$s_\bfi,t_\bfi,u_\bfi$, we denote by
$X_{\bfi}(s_\bfi,t_\bfi,u_\bfi)$ the Hermitian matrix in
$\bigotimes_{i\in [n]}M_{d_i}$ whose $({\bfi},{\bfi})$-th,
$(\bar{\bfi},\bar{\bfi})$-th and $({\bfi},\bar{\bfi})$-th entries
are given by $s_\bfi,t_\bfi$ and $u_\bfi$, respectively,
with zero entries otherwise. Then every multi-qubit {\sf
X}-shaped Hermitian matrix can be written by
\begin{equation}\label{notation}
X(s,t,u):=\sum_{{\bfi}\in B_0}X_{\bfi}(s_{\bfi}, t_{\bfi}, u_{\bfi}),
\end{equation}
for $s=\{s_\bfi:\bfi\in B_0\}$, $t=\{t_\bfi:\bfi\in B_0\}$ and $u=\{u_\bfi:\bfi\in B_0\}$,
where $B_0$ is the set of all indices beginning with $0$. If we endow the set $B_0$ with the lexicographic order,
and identify $B_0$ with $\{1,2,\dots, 2^{n-1}\}$ by the binary expansion
then this matrix (\ref{notation}) can be written as the following usual matrix
$$
\left(
\begin{matrix}
s_1 &&&&&&& u_1\\
& s_2 &&&&& u_2 & \\
&& \ddots &&& \iddots &&\\
&&& s_{2^{n-1}}&u_{2^{n-1}} &&&\\
&&& \bar u_{2^{n-1}}&t_{2^{n-1}}&&&\\
&& \iddots &&& \ddots &&\\
& \bar u_2 &&&&& t_2 &\\
\bar u_1 &&&&&&& t_1
\end{matrix}
\right).
$$

For a given {\sf X}-shaped Hermitian matrix $W=X(s,t,u)$,
we show in Section 3 that $\langle  \varrho,W \rangle \ge 0$ for any bi-separable state
$\varrho$ if and only if $\langle  \varrho,W \rangle \ge 0$ for any PPT mixture $\varrho$ if and only if the inequality
\begin{equation}\label{ineq_gew}
\sqrt{s_\bfi t_\bfi}+\sqrt{s_\bfj t_\bfj}\ge |u_\bfi|+|u_\bfj|
\end{equation}
holds for every choice of $\bfi,\bfj\in B_0$ with $\bfi\neq \bfj$.
We note that two diagonal entries $s_\bfi, t_\bfi$ of an {\sf X}-shaped genuine entanglement witness $W$
are allowed to be zero, even though the corresponding
anti-diagonal entries $u_\bfi, \bar{u}_\bfi$ are nonzero. This is the point why they are useful to detect genuine entanglement.

It was shown in \cite{{gao},{guhne10}} that if
an arbitrary multi-qubit state $\varrho$ whose diagonal and anti-diagonal parts are given by
$X(a,b,z)$ is bi-separable then the inequality
\begin{equation}\label{ineq_bi-sep}
\sum_{\bfj \ne \bfi} \sqrt {a_\bfj b_\bfj}\ge |z_\bfi|
\end{equation}
holds for each $\bfi\in B_0$. Using {\sf X}-shaped witnesses we
constructed, we see that this is necessary for PPT mixtures as well
as bi-separable states. We note that the inequality
(\ref{ineq_bi-sep}) is also known \cite{Rafsanjani} to be equivalent
to bi-separability for {\sf X}-shaped states.

We also find all optimal ones among {\sf X}-shaped genuine entanglement witnesses in Section 4.
An {\sf X}-shaped Hermitian $W=X(s,t,u)$ is an optimal genuine entanglement witness
if and only if it is a genuine entanglement witness with the spanning property if and only if
there exists an index $\bfi_0\in B_0$ such that
$s_{\bfi_0}=t_{\bfi_0}=0$, $|u_{\bfi_0}|=1$ and $s_\bfi t_\bfi=1$, $u_\bfi=0$ for $\bfi\neq \bfi_0$,
up to scalar multiplication.

In Section 5, we pay attention to decomposability of witnesses to show that a multi-qubit
{\sf X}-shaped witness $W=X(s,t,u)$ is decomposable if and only if
the inequality
\begin{equation}\label{ineq_decom}
\sum_{\bfi\in B_0}\sqrt {s_\bfi t_\bfi}\ge \sum_{\bfi\in B_0} |u_\bfi|
\end{equation}
holds. This shows that the notion of genuine entanglement witness is much stronger than decomposability
for {\sf X}-shaped witnesses. Indeed, only one pair $(s_\bfi,t_\bfi)$ of diagonal entries are allowed to be zero for
genuine entanglement witnesses, but all the diagonal entries except one pair of diagonal entries  may be
zero for decomposable matrices.
In the course of discussion, we also show that a multi-qubit {\sf X}-shaped state
$\varrho=X(a,b,z)$ is {\sl fully bi-separable}, that is, $S$-$T$ bi-separable for any bi-partition $S\sqcup T=[n]$
if and only if it is of PPT if and only if the inequality
\begin{equation}\label{ineq_ppt}
\sqrt{a_\bfi b_\bfi} \ge |z_\bfj|
\end{equation}
holds for every $\bfi,\bfj\in B_0$.


\section{{\sf X}-shaped multi-qubit genuine entanglement witnesses}

From now on, we restrict ourselves to the multi-qubit cases. So, $M_{d_i}$ will be the algebra $M_2$ of all
$2\times 2$ matrices, and an index will be a $\{0,1\}$-string.
We note that an {\sf X}-shaped matrix $W$ is positive, that is, positive semi-definite if and only if the following
$2\times 2$ matrix
$$
\left(\begin{matrix}
W_{{\bfi},{\bfi}} & W_{{\bfi},\bar{{\bfi}}}\\
W_{\bar{{\bfi}},{\bfi}} & W_{\bar{{\bfi}},\bar{{\bfi}}}
\end{matrix}\right)
$$
is positive for every ${\bfi} \in I_{[n]}$.

Now, we assume that $W$ is {\sf X}-shaped, and look for a condition with which the map
$\phi_W^{S,T}$ in (\ref{st_map}) is
positive. This map sends an element
$\sum_{{\bfi},{\bfj} \in  I_S} a_{{\bfi},{\bfj}}|{\bfi}\ran\lan {\bfj}|$  in $\bigotimes_{i\in S} M_{d_i}$ to
$\sum_{{\bfi},{\bfj} \in  I_S} a_{{\bfi},{\bfj}} W[{\bfi},{\bfj}] \in \bigotimes_{i\in T} M_{d_i}$,
which is
again {\sf X}-shaped, and so, $\sum_{{\bfi},{\bfj} \in  I_S} a_{{\bfi},{\bfj}} W[{\bfi},{\bfj}]$ is positive if and only if
$$
\sum_{{\bfi},{\bfj} \in  I_S} a_{{\bfi},{\bfj}}
\left(\begin{matrix}
W[{\bfi},{\bfj}]_{{\bfk},{\bfk}} &W[{\bfi},{\bfj}]_{{\bfk},\bar{\bfk}}\\
W[{\bfi},{\bfj}]_{\bar{\bfk},{\bfk}} &W[{\bfi},{\bfj}]_{\bar{\bfk},\bar{\bfk}}
\end{matrix}\right)\in M_2
$$
is positive
for each ${\bfk} \in I_T$. Therefore, we see that the map $\phi_W^{S,T}$ is positive if and only if
the map
$$
\Phi_{\bfk} : \sum_{{\bfi},{\bfj} \in  I_S} a_{{\bfi},{\bfj}}|{\bfi}\ran\lan {\bfj}|
\in \bigotimes_{i\in S} M_{d_i} \mapsto \sum_{{\bfi},{\bfj} \in  I_{S}} a_{{\bfi},{\bfj}}
\left(\begin{matrix}
W_{{\bfi} \diamond {\bfk}, {\bfj} \diamond {\bfk}} &W_{{\bfi} \diamond {\bfk}, {\bfj} \diamond {\bar{\bfk}}}\\
W_{{\bfi} \diamond {\bar{\bfk}}, {\bfj} \diamond {\bfk}} &W_{{\bfi} \diamond {\bar{\bfk}}, {\bfj} \diamond {\bar{\bfk}}}
\end{matrix}\right)\in M_2
$$
is positive for each ${\bfk} \in I_T$.
Now, we consider the transpose
$\Phi_{\bfk}^{\rm t} : M_2 \to \bigotimes_{i\in S} M_{d_i}$
of the map $\Phi_{\bfk}$, and the Choi matrix of $\Phi_{\bfk}^{\rm t}$.
For $p,q \in \{0,1\}$, we have
$$
\begin{aligned}
\Phi_{\bfk}^{\rm t} (|p\ran \lan q|)_{{\bfi},{\bfj}} & = \left\lan |p\ran \lan q|, \Phi_{\bfk} (|{\bfi}\ran \lan {\bfj}|) \right\ran \\
& = \left\lan |p\ran \lan q|, \left(\begin{matrix}
W_{{\bfi} \diamond {\bfk}, {\bfj} \diamond {\bfk}} &W_{{\bfi} \diamond {\bfk}, {\bfj} \diamond {\bar{\bfk}}}\\
W_{{\bfi} \diamond {\bar{\bfk}}, {\bfj} \diamond {\bfk}} &W_{{\bfi} \diamond {\bar{\bfk}},
{\bfj} \diamond {\bar{\bfk}}} \end{matrix}\right) \right\ran \\
& = \left\lan |p\ran \lan q|, \left(\begin{matrix}
W[{\bfk},{\bfk}]_{{\bfi},{\bfj}} &W[{\bfk},\bar{\bfk}]_{{\bfi},{\bfj}}\\
W[\bar{{\bfk}},{\bfk}]_{{\bfi},{\bfj}} &W[\bar{{\bfk}},\bar{\bfk}]_{{\bfi},{\bfj}}
\end{matrix}\right) \right\ran.
\end{aligned}
$$
Hence, the Choi matrix of $\Phi_{\bfk}^{\rm t}$ is given as
\begin{equation}\label{2->S}
W_{\bfk} :=
\left(\begin{matrix}
\Phi_{\bfk}^{\rm t}(|0\ran\lan0|) & \Phi_{\bfk}^{\rm t}(|0\ran\lan1|)\\
\Phi_{\bfk}^{\rm t}(|1\ran\lan0|) & \Phi_{\bfk}^{\rm t}(|1\ran\lan1|)
\end{matrix}\right)
=\left(\begin{matrix}
W[{\bfk},{\bfk}] &W[{\bfk},\bar{\bfk}]\\
W[\bar{{\bfk}},{\bfk}] &W[\bar{{\bfk}},\bar{\bfk}]
\end{matrix}\right)
\in M_2\left(\bigotimes_{i\in S} M_{d_i}\right).
\end{equation}
A bi-partite self-adjoint matrix is said to be block-positive if the pairing with any separable state is nonnegative. We recall that the positivity of a linear map is equivalent to block-positivity of its Choi matrix, and so,
we have the following:

\begin{lemma}\label{2->Sprop}
Suppose that $W$ is an {\sf X}-shaped matrix in $\bigotimes_{i\in [n]} M_{d_i}$ with $d_i=2$. For a bi-partition $S\sqcup T=[n]$,
the following are equivalent:
\begin{enumerate}
\item[(i)]
the map $\phi_W^{S,T} : \bigotimes_{i\in S} M_{d_i} \to \bigotimes_{i\in T} M_{d_i}$ is positive;
\item[(ii)]
the matrix $W_{\bfk}$ given as {\rm (\ref{2->S})} is block positive in $M_2\left(\bigotimes_{i\in S} M_{d_i}\right)$, or equivalently
the map $\Phi_{\bfk}^{\rm t} : M_2 \to \bigotimes_{i\in S} M_{d_i}$ is positive, for each ${\bfk} \in I_T$.
\end{enumerate}
\end{lemma}

Since $W_{\bfk}$ is again {\sf X}-shaped, we can apply Lemma \ref{2->Sprop} to the matrix
$W_{\bfk} \in M_2 \otimes (\bigotimes_{i\in S} M_{d_i})$ with the given bi-partition.
Note that the map (\ref{st_map}) associated with $W_{\bfk}$ coincides with
$\Phi_{\bfk}^{\rm t}:M_2\to \bigotimes_{i\in S} M_{d_i}$ by (\ref{2->S}). Therefore, the map
$\Phi_{\bfk}^{\rm t}$ is positive if and only if for each ${\bfi} \in I_S$ the matrix
\begin{equation}\label{ki}
\begin{aligned}
(W_{\bfk})_{\bfi} = &
\left(\begin{matrix}
W_{\bfk}[{\bfi},{\bfi}] & W_{\bfk}[{\bfi},\bar{\bfi}] \\
W_{\bfk}[\bar{\bfi},{\bfi}] & W_{\bfk}[\bar{\bfi},\bar{\bfi}]
\end{matrix}\right)
\\
= & \left(\begin{matrix}
W_{{\bfi} \diamond {\bfk}, {\bfi} \diamond {\bfk}}
& W_{{\bfi} \diamond {\bfk}, {\bfi} \diamond {\bar{\bfk}}}
&W_{{\bfi} \diamond {\bfk}, {\bar{\bfi}} \diamond {\bfk}}
& W_{{\bfi} \diamond {\bfk}, {\bar{\bfi}} \diamond {\bar{\bfk}}} \\
W_{{\bfi} \diamond {\bar{\bfk}}, {\bfi} \diamond {\bfk}}
& W_{{\bfi} \diamond {\bar{\bfk}}, {\bfi} \diamond {\bar{\bfk}}}
&W_{{\bfi} \diamond {\bar{\bfk}}, {\bar{\bfi}} \diamond {\bfk}}
& W_{{\bfi} \diamond {\bar{\bfk}}, {\bar{\bfi}} \diamond {\bar{\bfk}}} \\
W_{{\bar{\bfi}} \diamond {\bfk}, {\bfi} \diamond {\bfk}}
& W_{{\bar{\bfi}} \diamond {\bfk}, {\bfi} \diamond {\bar{\bfk}}}
&W_{{\bar{\bfi}} \diamond {\bfk}, {\bar{\bfi}} \diamond {\bfk}}
& W_{{\bar{\bfi}} \diamond {\bfk}, {\bar{\bfi}} \diamond {\bar{\bfk}}} \\
W_{{\bar{\bfi}} \diamond {\bar{\bfk}}, {\bfi} \diamond {\bfk}}
& W_{{\bar{\bfi}} \diamond {\bar{\bfk}}, {\bfi} \diamond {\bar{\bfk}}}
&W_{{\bar{\bfi}} \diamond {\bar{\bfk}}, {\bar{\bfi}} \diamond {\bfk}}
& W_{{\bar{\bfi}} \diamond {\bar{\bfk}}, {\bar{\bfi}} \diamond {\bar{\bfk}}}
\end{matrix}\right)\\
= & \left(\begin{matrix}
W_{{\bfi} \diamond {\bfk}, {\bfi} \diamond {\bfk}}
& \cdot
& \cdot
& W_{{\bfi} \diamond {\bfk}, {\bar{\bfi}} \diamond {\bar{\bfk}}} \\
\cdot
& W_{{\bfi} \diamond {\bar{\bfk}}, {\bfi} \diamond {\bar{\bfk}}}
&W_{{\bfi} \diamond {\bar{\bfk}}, {\bar{\bfi}} \diamond {\bfk}}
& \cdot \\
\cdot
& W_{{\bar{\bfi}} \diamond {\bfk}, {\bfi} \diamond {\bar{\bfk}}}
&W_{{\bar{\bfi}} \diamond {\bfk}, {\bar{\bfi}} \diamond {\bfk}}
& \cdot \\
W_{{\bar{\bfi}} \diamond {\bar{\bfk}}, {\bfi} \diamond {\bfk}}
& \cdot
& \cdot
& W_{{\bar{\bfi}} \diamond {\bar{\bfk}}, {\bar{\bfi}} \diamond {\bar{\bfk}}}
\end{matrix}\right) \in M_2(M_2)
\end{aligned}
\end{equation}
is block positive. Therefore, we have the following:

\begin{theorem}\label{block}
Suppose that $W$ is an {\sf X}-shaped multi-qubit Hermitian matrix.
For a bi-partition $S\sqcup T=[n]$,
the following are equivalent:
\begin{enumerate}
\item[(i)]
the map $\phi_W^{S,T}$ is positive;
\item[(ii)]
the matrix {\rm (\ref{ki})} is block positive for every ${\bfi} \in I_S$ and ${\bfk} \in I_T$.
\end{enumerate}
\end{theorem}

For a given bi-partition $[n]=S\sqcup T$ and ${\bfi}\in I_S, {\bfk}\in I_T$, we see that
$\bfi \diamond\bfk  \neq\bfi \diamond \bar{\bfk}$,
$\bfi \diamond\bfk  \neq\bar{\bfi}\diamond \bfk $ and the matrix (\ref{ki}) is of the form
\begin{equation}\label{2x2_matrix}
\left(\begin{matrix}
W_{{\bfi},{\bfi}} & \cdot & \cdot & W_{{\bfi},\bar{\bfi}} \\
\cdot & W_{{\bfj},{\bfj}} & W_{{\bfj},\bar{\bfj}} & \cdot \\
\cdot & W_{\bar{\bfj},{\bfj}} & W_{\bar{\bfj},\bar{\bfj}} & \cdot \\
W_{\bar{\bfi},{\bfi}} & \cdot & \cdot & W_{\bar{\bfi},\bar{\bfi}}
\end{matrix}\right) \in M_2(M_2)
\end{equation}
with indices $\bfi , \bfj $ on $[n]$ satisfying $\bfi \neq \bfj $ and $\bfi \neq \bar{\bfj}$.
Conversely, if $\bfi $ and $\bfj $ are indices on $[n]$ with ${\bfj} \ne {\bfi}, \bar{\bfi}$, then we put
$$
S=\{i\in [n]: {\bfi}(i)= {\bfj}(i)\} \quad \text{and} \quad T=\{i\in [n]: {\bfi}(i)\ne {\bfj}(i)\}.
$$
Then, $S \sqcup T$ is a bi-partition of $[n]$, and we have
${\bfi}={\bfi}|_S \diamond {\bfi}|_T$ and
${\bfj}={\bfj}|_S \diamond {\bfj}|_T = {\bfi}|_S \diamond \overline{{\bfi}|}_{S^c}$.
Therefore, we have the relations (i) $\Longleftrightarrow$ (ii) $\Longleftrightarrow$ (v) in
Theorem \ref{main} below.

In Theorem \ref{main}, we will also show that a non-positive {\sf X}-shaped multi-qubit
Hermitian matrix is a genuine entanglement witness if and only if
 $\langle W, \varrho \rangle \ge 0$ for any PPT mixture $\varrho$. In order to discuss this part, we need the notion of
 partial transposes for multi-partite systems.
 For a given subset $S \subset \{1,2,\cdots,n\}$, the partial transpose $T(S)$ on $\bigotimes_{i\in [n]} M_{d_i}$
is the linear map satisfying
\begin{equation}\label{par-transpose}
(a_1\ot a_2\ot\cdots\ot a_n)^{T(S)}:=b_1\ot b_2\ot\cdots\ot b_n,
\quad \text{\rm with}\ b_i=\begin{cases} a_i^\ttt, &i\in S,\\ a_i,
&i \notin S,\end{cases}
\end{equation}
where $a^\ttt$ denotes the transpose of the matrix $a$.
For an index $\bfi$ on $[n]$ and a subset $S$ of $[n]$, we also define the index $\bar{\bfi}^S$ by
\begin{equation}\label{bar_s}
\bar{\bfi}^S(i)=
\begin{cases}
i+1\mod 2, & i\in S, \\
i, & i \notin S.
\end{cases}
\end{equation}
We note that $\bar\bfi$ is nothing but $\bar{\bfi}^{[n]}$ with this definition.
We have the relations
\begin{equation}\label{pt_unit}
|{\bfi}\ran\lan {\bfi}|^{T(S)}=
|{\bfi}\ran\lan {\bfi}|,
\qquad
|{\bfi}\ran\lan \bar{\bfi}|^{T(S)}=
|\bar{\bfi}^S \ran\lan \bar{\bfi}^{S^c}|.
\end{equation}
If we write ${\bfj}=\bar{\bfi}^{S}$ then
$T(S)$ sends $|{\bfi}\rangle\langle \bar{\bfi}|$ to
 $|{\bfj}\rangle\langle \bar{\bfj}|$.

\begin{theorem}\label{main}
Suppose that $W=X(s,t,u)$ is an {\sf X}-shaped multi-qubit Hermitian matrix {\rm (\ref{notation})}
with nonnegative diagonals.
Then the following are equivalent:
\begin{enumerate}
\item[(i)]
$\lan\varrho, W\ran\ge 0$ for every $n$ qubit bi-separable state $\varrho$;
\item[(ii)]
the map $\phi_W^{S,T}$ is positive for any bi-partition $S\sqcup T=[n]$;
\item[(iii)]
$\lan\varrho, W\ran\ge 0$ for every $n$ qubit PPT mixture $\varrho$;
\item[(iv)]
for every nontrivial subset $S$ of $[n]$, there are positive $P$ and $Q$ such that $W=P+Q^{T(S)}$;
\item[(v)]
the matrix {\rm (\ref{2x2_matrix})}
is block positive for every indices $\bfi , \bfj$ with  ${\bfj} \ne {\bfi}, \bar{\bfi}$;
\item[(vi)]
the inequality {\rm (\ref{ineq_gew})}
holds for every indices $\bfi , \bfj \in B_0$ with  ${\bfi} \ne {\bfj}$.
\end{enumerate}
\end{theorem}

\begin{proof}
We first note that the matrix (\ref{2x2_matrix}) can be written by
$$
\left(\begin{matrix}
s_{{\bfi},{\bfi}} & \cdot & \cdot & u_{{\bfi},\bar{\bfi}} \\
\cdot & s_{{\bfj},{\bfj}} & u_{{\bfj},\bar{\bfj}} & \cdot \\
\cdot & \bar u_{{\bfi},\bar{\bfi}} & t_{\bar{\bfj},\bar{\bfj}} & \cdot \\
\bar u_{{\bfj},\bar{\bfj}} & \cdot & \cdot & t_{\bar{\bfi},\bar{\bfi}}
\end{matrix}\right)
$$
for $\bfi,\bfj\in B_0$.
Consider the linear map between $M_2$ whose Choi matrix is given by this. Then the positivity of this map
is equivalent to the inequality (\ref{ineq_gew}) by \cite[Lemma 6.1]{han_kye_tri}.
Therefore, we see that the statements (i), (ii), (v) and (vi) are equivalent.
The directions (iv) $\Longrightarrow$ (iii) $\Longrightarrow$ (i) follow from the standard duality.
We complete the proof by proving the direction (vi) $\Longrightarrow$ (iv).

Suppose that (vi) holds.
If the inequality  $\sqrt{s_\bfi t_\bfi}\ge |u_\bfi|$
holds for every index $\bfi \in B_0$, then $W$ is positive, and so there is nothing to prove.
In the other case, there exists a unique index $\bfi \in B_0$ such that
$\sqrt{s_\bfi t_\bfi}< |u_\bfi|$,
by the inequality (\ref{ineq_gew}).
To prove (iv), we may assume that $1\notin S$ because $Q^{T(S^c)}=(Q^\ttt)^{T(S)}$.
For a given subset $S$ with $1\notin S$ and ${\bfi} \in B_0$, put ${\bfj}=\bar{\bfi}^{S}$.
Then we see that $\bfi,\bfj\in B_0$, and
the matrix
$$
D:=X_\bfi(s_\bfi, t_\bfi, u_\bfi)+X_\bfj(s_\bfj, t_\bfj, u_\bfj)
$$
looks like (\ref{2x2_matrix}) without changing the size of $W$.

We first note that $W-D$ is positive. By the inequality (\ref{ineq_gew}),
the matrix $D$ is essentially a block positive matrix in $M_2(M_2)$
if we ignore zero entries,
and so $D=P_0+Q^\tau$ with positive $P_0$ and $Q$
when it is considered as a matrix in $M_2(M_2)$, where $Q^\tau$ is the partial transpose of $Q$
with respect to the second subsystem.
But, $Q^\tau$ is nothing but $Q^{T(S)}$ if $Q$ is considered as a matrix in $\bigotimes_{i\in [n]} M_{d_i}$.
The proof is complete by putting $P=W-D+P_0$.
\end{proof}

\begin{corollary}\label{hhhhhh}
Let $\varrho$ be a multi-qubit state whose diagonal and anti-diagonal parts are given by
$X(a,b,z)$. If $\varrho$ is a PPT mixture then the inequality {\rm (\ref{ineq_bi-sep})}
holds for each $\bfi\in B_0$.
\end{corollary}

\begin{proof}
We first consider the case when all of $a_{\bfj}$ and $b_{\bfj}$ are nonzero.
We consider the witness $W$ defined by
$$
W=\sum_{{\bfj}\in B_0\setminus \{{\bfi}\}}
X_{\bfj}\left(\sqrt {b_{\bfj} \over a_{\bfj}},\ \sqrt {a_{\bfj} \over b_{\bfj}},\ 0\right)
+X_{\bfi}\left( 0,\ 0,\ -e^{-i\theta_{\bfi}}\right),
$$
where $\theta_{\bfi}$ is the argument of $z_{\bfi}$.
Since $W$ satisfies the inequality (\ref{ineq_gew}), we have
$$
0 \le {1 \over 2} \lan \varrho, W \ran =
\sum_{{\bfj}\in B_0\setminus \{{\bfi}\}}\sqrt{a_{\bfj} b_{\bfj}}- |z_{\bfi}|.
$$
If some of  $a_{\bfj}$ and $b_{\bfj}$ are zero, then we consider the (unnormalized) state $\varrho+\varepsilon I$ with the identity
matrix $I$. Since $\varrho+\varepsilon I$ is still a PPT mixture, we may apply the same argument as above to get the corresponding inequality.
This completes the proof by letting $\varepsilon\to 0$.
\end{proof}

Because the inequality (\ref{ineq_bi-sep}) is also known \cite{Rafsanjani} to be equivalent
to bi-separability for {\sf X}-shaped states, we have the following:

\begin{corollary}\label{esrdh}
For an {\sf X}-shaped multi-qubit state $\varrho=X(a,b,z)$, the following are equivalent:
\begin{enumerate}
\item[(i)]
$\varrho$ is bi-separable;
\item[(ii)]
$\varrho$ is a PPT mixture;
\item[(iii)]
the inequality
{\rm (\ref{ineq_bi-sep})} holds for every ${\bfi}\in B_0$.
\end{enumerate}
\end{corollary}

We note that the equivalence between (i) and (ii) in Corollary \ref{esrdh} also follows directly from Proposition
\ref{ST_sep}.

\section{Optimal genuine entanglement witnesses}

For a given genuine entanglement witness $W$, we consider the set
$G_W$ of all genuine entanglement $\varrho$ which are detected by $W$
in the sense of $\lan\varrho, W\ran <0$. Following \cite{lew00}, we say that $W$ is {\sl optimal} if the set $G_W$ is maximal.
If $P$ is positive then we have the relation $G_{W+P}\subset G_W$, and $W$ is not optimal if there is a
nonzero positive matrix $P$
such that $W-P$ is still a genuine entanglement witness.
It is very difficult in general to determine if a given witness is optimal or not.
The notion of the spanning property is stronger than the optimality, and easier to check.

We say that a vector $|z\ran\in\bigotimes_{i\in[n]}\mathbb C^{d_i}$ is a {\sl bi-product vector}
if there is a bi-partition $S\sqcup T=[n]$ such that $|z\ran$ is a product vector as an element of
$\left(\bigotimes_{i\in S}\mathbb C^{d_i}\right)\otimes \left(\bigotimes_{i\in T}\mathbb C^{d_i}\right)$.
For a given genuine entanglement witness $W$, we denote by $P_W$ the set of all bi-product vectors $|z\ran$ such that
$$
\langle \bar z | W | \bar z\rangle =\langle |z\ran\lan z|,\, W\rangle =0.
$$
We say that $W$ has the {\sl spanning property} if the set $P_W$ spans the whole space $\bigotimes_{i\in[n]}\mathbb C^{d_i}$.
By the same argument as in \cite{lew00}, we see that the spanning property implies the optimality. It is important to note that
if $W$ has the spanning property then set $G_W$ has nonempty interior, and so has non-zero volume, by the same argument
as in \cite{ha+kye_exposed}. See also \cite{kye_ritsu}.

\begin{theorem}
Suppose that $W=X(s,t,u)$ is an {\sf X}-shaped $n$-qubit genuine entanglement witness of the form {\rm (\ref{notation})}.
Then the following are equivalent:
\begin{enumerate}
\item[(i)]
$W$ is an optimal genuine entanglement witness;
\item[(ii)]
$W$ is a genuine entanglement witness with the spanning property;
\item[(iii)]
There exists ${\bfi}_0\in B_0$ and positive number $r>0$ with the properties:
\begin{itemize}
\item
$s_{{\bfi}_0}=t_{{\bfi}_0}=0$ and $|u_{{\bfi}_0}|=r$,
\item
$\sqrt{s_{\bfi}t_{\bfi}}=r$ and $u_{\bfi}= 0$ for ${\bfi}\neq {\bfi}_0$.
\end{itemize}
\end{enumerate}
\end{theorem}

\begin{proof}
(i) $\Longrightarrow$ (iii).
If $W$ is positive then it is never a genuine entanglement witness. So, there exists a unique ${\bfi}_0\in B_0$ such that
$\sqrt{ s_{{\bfi}_0} t_{{\bfi}_0}}<|u_{{\bfi}_0}|$ and $s_\bfi t_\bfi\neq 0$ for $\bfi\neq \bfi_0$ by (\ref{ineq_gew}). Put
$$
P_1=\sum_{{\bfi}\in B_0, {\bfi}\neq {\bfi}_0}
X_{\bfi}\left( \sqrt{s_{\bfi} \over t_{\bfi}}|u_{\bfi}|,\ \sqrt{t_{\bfi} \over s_{\bfi}}|u_{\bfi}|,\ u_{\bfi}\right)
$$
which is positive. From Theorem \ref{main}, we see that $W-P_1$ is a genuine entanglement witness, by the inequality
$$
\sqrt{s_{{\bfi}_0} t_{{\bfi}_0}} + \sqrt{\left(s_{\bfi} - \sqrt{s_{\bfi} \over t_{\bfi}}|u_{\bfi}|\right)
\left(t_{\bfi} - \sqrt{t_{\bfi} \over s_{\bfi}}|u_{\bfi}|\right)}
 = \sqrt{s_{{\bfi}_0} t_{{\bfi}_0}} + \sqrt{s_{\bfi} t_{\bfi}} - |u_{\bfi}|
 \ge |u_{{\bfi}_0}|
$$
for each ${\bfi}\in B_0$ with ${\bfi}\neq {\bfi}_0$. Therefore, $P_1$ must be zero, and it follows that
$u_{\bfi}= 0$ whenever ${\bfi}\neq {\bfi}_0$. We write $u_{{\bfi}_0}=re^{i\theta}$.

Let $R$ be the minimum
of $\sqrt{s_{\bfi}t_{\bfi}}$ through ${\bfi}\in B_0\setminus\{ {\bfi}_0\}$. If $R\ge r$ then
$$
P_2=W_{{\bfi}_0}\left( s_{{\bfi}_0}, t_{{\bfi}_0}, 0\right)
$$
is positive, and $W-P_2$ is a genuine entanglement witness by Theorem \ref{main} again. Therefore, we have
$s_{{\bfi}_0}=t_{{\bfi}_0}=0$.
If $R<r$, then put
$$
P_3=W_{{\bfi}_0}\left( s_{{\bfi}_0}, t_{{\bfi}_0}, (r-R)e^{i\theta}\right).
$$
Then, we see again that $P_3$ is positive, and $W-P_3$ is a genuine entanglement witness.
Therefore, we also have $s_{{\bfi}_0}=t_{{\bfi}_0}=0$ in any cases, as it was desired.
Since $W$ is a genuine entanglement witness, we have $\sqrt{s_{\bfi}t_{\bfi}}\ge r$ for each
${\bfi}\in B_0 \setminus\{ {\bfi}_0\}$.
If $\sqrt{s_{\bfi}t_{\bfi}}>r$ for some ${\bfi}$, then we can subtract nonzero
diagonal matrix from $W$. Therefore, we have $\sqrt{s_{\bfi}t_{\bfi}}=r$ for each ${\bfi}\in B_0 \setminus\{ {\bfi}_0\}$.

Now, it remains to show the implication (iii) $\Longrightarrow$ (ii). To do this, we use the notations
${\bfzero}_S$ and ${\bfone}_S$ for indices on $S$ which are constant functions with values $0$ and $1$, respectively.
Suppose that (iii) holds. We may assume that ${\bfi}_0={\bfzero}_{[n]}$ and $r=1$, without loss of generality.
For a bi-partition $S\sqcup T=[n]$ and $\alpha\in\mathbb C$, we define two vectors
$|x_S(\alpha)\ran\in \bigotimes_{i\in S}\mathbb C^{d_i}$ and $|y_T(\alpha)\ran\in \bigotimes_{i\in T}\mathbb C^{d_i}$
by
$$
|x_S(\alpha)\ran = |{\bfzero}_S\ran +\alpha |{\bfone}_S\ran,\qquad
|y_T(\alpha)\ran = s\bar u|{\bfzero}_T\ran -\bar\alpha |{\bfone}_T\ran,
$$
with $s=W_{{\bfzero}_S\diamond {\bfone}_T, {\bfzero}_S\diamond {\bfone}_T}$ and
$u=W_{{\bfzero}_{[0]}, {\bfone}_{[0]}}$. Put
$$
\begin{aligned}
|z_{ST}(\alpha)\ran
:&=|x_S(\alpha)\ran\otimes|y_T(\alpha)\ran\\
&=
s\bar u|{\bfzero}_S\diamond {\bfzero}_T\ran
-\bar \alpha |{\bfzero}_S\diamond {\bfone}_T\ran
+s\bar u\alpha |{\bfone}_S\diamond {\bfzero}_T\ran
-|\alpha|^2|{\bfone}_S\diamond {\bfone}_T\ran.
\end{aligned}
$$
It is straightforward to see that $|z_{ST}(\alpha)\ran \in P_W$ and the set $\{|z_{ST}(\alpha)\ran: \alpha\in\mathbb C\}$
spans the $4$-dimensional space $V_{ST}$ spanned by $|{\bfzero}_S\diamond {\bfzero}_T\ran$,
$|{\bfzero}_S\diamond {\bfone}_T\ran$, $|{\bfone}_S\diamond {\bfzero}_T\ran$ and
$|{\bfone}_S\diamond {\bfone}_T\ran$. Now, it is clear that the span of $V_{ST}$ through bi-partition $S\sqcup T=[n]$
coincides with the whole space $\bigotimes_{i\in[n]}\mathbb C^{d_i}$.
\end{proof}

A typical example of three qubit optimal genuine entanglement witness is given by
$$
\left(
\begin{matrix}
\cdot &\cdot &\cdot &\cdot &\cdot &\cdot &\cdot &e^{i\theta} \\
\cdot &s_2 &\cdot &\cdot &\cdot &\cdot &\cdot &\cdot \\
\cdot &\cdot &s_3 &\cdot &\cdot &\cdot &\cdot &\cdot \\
\cdot &\cdot &\cdot &s_4 &\cdot &\cdot &\cdot &\cdot \\
\cdot &\cdot &\cdot &\cdot &1 \slash s_4 &\cdot &\cdot &\cdot \\
\cdot &\cdot &\cdot &\cdot &\cdot &1\slash s_3 &\cdot &\cdot \\
\cdot &\cdot &\cdot &\cdot &\cdot &\cdot &1\slash s_2  &\cdot \\
e^{-i\theta} &\cdot &\cdot &\cdot &\cdot &\cdot &\cdot &\cdot
\end{matrix}
\right),
$$
where $\cdot$ denotes zero.
We see that these witnesses detect all the GHZ type pure state \cite{abls},
as it was discussed in \cite{han_kye_tri}.


\section{Decomposability of {\sf X}-shaped multi-qubit witnesses}

In this section, we characterize the decomposability of {\sf X}-shaped multi-qubit witnesses,
in terms of the entries.
We recall the definition (\ref{par-transpose}) of the partial transpose $T(S)$
for a given subset $S\subset [n]$.
A state $\varrho$ in $\bigotimes_{i\in[n]} M_{d_i}$
is said to be of PPT (positive partial transpose) if $\varrho^{T(S)}$ is
positive for every subset $S$ of $[n]$. Therefore, $\varrho$ is of PPT if and only if it is in the
intersection of convex cones
$$
\mathbb T^S:=\{A\in\bigotimes_{i\in [n]} M_{d_i} : A^{T(S)}  \ {\text{\rm is positive}}\}
$$
through subsets $S$ of $[n]$. On the other hand, a Hermitian matrix $D$ is said to be {\sl decomposable} if
it is in the convex hull of the convex cones $\mathbb T^S$ through subsets $S$ of $[n]$.

It is easy to see that the convex cone $\mathbb D$ of all decomposable matrices in $\bigotimes_{i\in[n]} M_{d_i}$ is closed
by Caratheodory's theorem \cite[Theorem 17.2]{R}, which tells us that the convex hull of a compact set is again compact.
Therefore, we can apply the duality between the convex hulls and the intersections to get the following:

\begin{proposition}\label{dual-PPT}
For a state $\varrho$ and a Hermitian $W$ in $\bigotimes_{i\in[n]} M_{d_i}$, we have the following:
\begin{enumerate}
\item[(i)]
$\varrho$ is a PPT state if and only if $\lan\varrho, W\ran\ge 0$ for each decomposable $W$.
\item[(ii)]
$W$ is decomposable if and only if $\lan\varrho, W\ran\ge 0$ for each PPT state $\varrho$.
\end{enumerate}
\end{proposition}

\begin{proposition}\label{ST_sep}
Suppose that an {\sf X}-shaped $n$-qubit state $\varrho=X(a,b,z)$ and a bi-partition $[n]=S\sqcup T$ are given.
If $1\notin S$ then the following are equivalent:
\begin{enumerate}
\item[(i)]
$\varrho$ is $S$-$T$ bi-separable;
\item[(ii)]
$\varrho$ is $S$-$T$ PPT.
\item[(iii)]
the inequality {\rm (\ref{ineq_ppt})}
holds for every ${\bfi}, {\bfj}\in B_0$ with ${\bfi}=\bar{\bfj}^S$.
\end{enumerate}
\end{proposition}

\begin{proof}
The equivalence (ii) $\Longleftrightarrow$ (iii) follows from the relation
$$
X_{\bfi}(a_{\bfi}, b_{\bfi}, z_{\bfi})^{T(S)}
=X_{\bfi}(a_{\bfi}, b_{\bfi}, 0)
+X_{\bar{\bfi}^S}(0,0, z_{\bfi}),
$$
by (\ref{pt_unit}).
It remains to prove (ii) $\Longrightarrow$ (i).
If we write
$$
\varrho={1 \over 2}\sum_{{\bfi}\in B_0}\left(
X_{\bfi}(a_{\bfi}, b_{\bfi}, z_{\bfi}) +
X_{\bar{\bfi}^S} (a_{\bar{\bfi}^S}, b_{\bar{\bfi}^S}, z_{\bar{\bfi}^S})\right),
$$
then every summand is a two qubit PPT state if we ignore zero entries.
Therefore, it is separable as a two qubit state which is $S$-$T$ separable as an $n$-qubit state.
\end{proof}

Since we may interchange the roles of $S$ and $T$ in the bi-partition $[n]=S\sqcup T$,
the assumption $1\notin S$ in Proposition \ref{ST_sep} is actually superfluous.
We put this assumption just because the statement (iii)
makes sense only when $1\notin S$.
Taking a partial transpose of an {\sf X}-shaped multi-qubit matrix is nothing but a rearrangement of anti-diagonal entries,
while fixing diagonal entries.
Conversely, an anti-diagonal entries of an {\sf X}-shaped multi-qubit matrix $W$ may be moved
to any other anti-diagonal places by a suitable operation of partial transpose.
To see this, for given two different indices ${\bfi}, {\bfj}$ in $B_0$, we put
$S=\{i\in [n]: {\bfi}(i)\neq {\bfj}(i)\}$.
Then, we have
\begin{equation}\label{anti-PT}
X_{\bfi}(0,0,z)^{T(S)}=X_{\bfj}(0,0,z).
\end{equation}
Therefore, we have the following:

\begin{theorem}\label{X-PPT}
Let $\varrho=X(a,b,z)$ be an {\sf X}-shaped $n$-qubit state.
Then the following are equivalent:
\begin{enumerate}
\item[(i)]
$\varrho$ is fully bi-separable; 
\item[(ii)]
$\varrho$ is PPT;
\item[(iii)] the inequality
{\rm (\ref{ineq_ppt})}
holds for every choice of ${\bfi},{\bfj}\in B_0$.
\end{enumerate}
\end{theorem}

\begin{corollary}\label{oiuou}
Every {\sf X}-shaped multi-qubit PPT state is bi-separable.
\end{corollary}

Now, we turn our attention to the decomposability of {\sf X}-shaped entanglement witnesses.

\begin{theorem}\label{decom}
For an {\sf X}-shaped $n$-qubit Hermitian $W=X(s,t,u)$ with nonnegative diagonals, the following are equivalent:
\begin{enumerate}
\item[(i)]
$W$ is decomposable;
\item[(ii)]
$\lan \varrho,W\ran\ge 0$ for every PPT state $\varrho$;
\item[(iii)]
$\lan \varrho,W\ran\ge 0$ for every fully bi-separable state $\varrho$;
\item[(iv)]
the inequality {\rm (\ref{ineq_decom})} holds.
\end{enumerate}
\end{theorem}

\begin{proof}
Equivalence between (i) and (ii) is a part of Proposition \ref{dual-PPT}.
Since every fully bi-separable state is of PPT, we also have the implication (ii) $\Longrightarrow$ (iii).
To prove the directions (iii) $\Longrightarrow$ (iv) $\Longrightarrow$ (i),
we first consider the case when $s_{\bfi}, t_{\bfi} >0$ for every ${\bfi}\in B_0$.
We also write $u_{\bfi} = |u_{\bfi}| e^{{ i} \theta_{\bfi}}$.
We consider the {\sf X}-shaped state
$$
\varrho=\sum_{{\bfi}\in B_0} X_{\bfi}
\left(\sqrt{t_{\bfi} \over s_{\bfi}},\ \sqrt{s_{\bfi} \over t_{\bfi}},\ -e^{-i\theta_{\bfi}}\right).
$$
This is a fully bi-separable state by Theorem \ref{X-PPT}, and so we have
$$
0 \le {1 \over 2} \langle W, \varrho \rangle = \sum_{{\bfi}\in B_0} \sqrt{s_{\bfi} t_{\bfi}} - \sum_{{\bfi}\in B_0} |u_{\bfi}|.
$$
This completes the proof of (iii) $\Longrightarrow$ (iv).

For the direction (iv) $\Longrightarrow$ (i), we suppose that the inequality (\ref{ineq_decom}) holds.
Put
$$
S_+=\{{\bfi}\in B_0: \sqrt{s_{\bfi}t_{\bfi}}\ge |u_{\bfi}|\},\qquad
S_-=\{{\bfj}\in B_0: \sqrt{s_{\bfj}t_{\bfj}}< |u_{\bfj}|\},
$$
and define
$$
\begin{aligned}
W^{+}
&=\sum_{{\bfi}\in S_+}X_{\bfi}\left(\frac{|u_{\bfi}|s_{\bfi}}{\sqrt {s_{\bfi}t_{\bfi}}},\
   \frac{|u_{\bfi}|t_{\bfi}}{\sqrt {s_{\bfi}t_{\bfi}}},\ u_{\bfi}\right)
  +\sum_{{\bfj}\in S_-}X_{\bfj}(s_{\bfj},t_{\bfj}, \sqrt{s_{\bfj}t_{\bfj}}e^{{ i} \theta_{\bfj}}), \\
W^{-}
&=\sum_{{\bfi}\in S_+}X_{\bfi}\left( \left(1-\frac{|u_{\bfi}|}{\sqrt {s_{\bfi}t_{\bfi}}}\right)s_{\bfi},\
   \left(1-\frac{|u_{\bfi}|}{\sqrt {s_{\bfi}t_{\bfi}}}\right)t_{\bfi},0\right)
  +\sum_{{\bfj}\in S_-}X_{\bfj}(0, 0, (|u_{\bfj}|-\sqrt{s_{\bfj}t_{\bfj}})e^{{ i}\theta_{\bfj}}).
\end{aligned}
$$
Then $W^+$ is positive and $W=W^+ +W^-$.
For the brevity, we write
$$
p_{\bfi}=\left(1-\frac{|u_{\bfi}|}{\sqrt {s_{\bfi}t_{\bfi}}}\right)s_{\bfi},\qquad
q_{\bfi}=\left(1-\frac{|u_{\bfi}|}{\sqrt {s_{\bfi}t_{\bfi}}}\right)t_{\bfi},\qquad
v_{\bfj}=(|u_{\bfj}|-\sqrt{s_{\bfj}t_{\bfj}})e^{{ i}\theta_{\bfj}},
$$
so that
$$
W^{-}=\sum_{{\bfi}\in S_+}X_{\bfi}(p_{\bfi},q_{\bfi},0)+\sum_{{\bfj}\in S_-}X_{\bfj}(0,0,v_{\bfj}).
$$
By the inequality
{\rm (\ref{ineq_decom})},
we have
$$
\sum_{{\bfi}\in S_+}\sqrt{p_{\bfi}q_{\bfi}}
 = \sum_{{\bfi}\in S_+} \left(1-\frac{|u_{\bfi}|}{\sqrt {s_{\bfi}t_{\bfi}}}\right) \sqrt{s_{\bfi} t_{\bfi}}
 = \sum_{{\bfi}\in S_+}\sqrt{s_{\bfi} t_{\bfi}} -|u_{\bfi}|
 \ge \sum_{{\bfj}\in S_-} |u_{\bfj}|-\sqrt{s_{\bfj} t_{\bfj}} 
 = \sum_{{\bfj}\in S_-}|v_{\bfj}|.
$$

Now, we put
$$
c_{\bfi}= \frac {\sqrt{p_{\bfi}q_{\bfi}}} {\sum_{{\bfk}\in S_+}\sqrt{p_{\bfk}q_{\bfk}}} ,\qquad
d_{\bfj}= \frac {|v_{\bfj}|} {\sum_{{\bfk}\in S_-}|v_{\bfk}|}
$$
for each ${\bfi}\in S_+$ and ${\bfj}\in S_-$, and define
$$
W^{{\bfi},{\bfj}}=X_{\bfi}(d_{\bfj}p_{\bfi},d_{\bfj}q_{\bfi},0)+X_{\bfj}(0,0,c_{\bfi}v_{\bfj}).
$$
We decompose $W^-$ as
$$
\begin{aligned}
W^{-} & = \sum_{{\bfi}\in S_+}X_{\bfi}(p_{\bfi},q_{\bfi},0)+\sum_{{\bfj}\in S_-}X_{\bfj}(0,0,v_{\bfj}) \\
& = \sum_{{\bfi}\in S_+} \sum_{{\bfj}\in S_-} X_{\bfi}(d_{\bfj} p_{\bfi},d_{\bfj} q_{\bfi},0)
   +\sum_{{\bfj}\in S_-} \sum_{{\bfi}\in S_+} X_{\bfj}(0,0,c_{\bfi} v_{\bfj})
 = \sum_{j\in S_-}\sum_{i\in S_+}W^{{\bfi},{\bfj}}.
\end{aligned}
$$
By the relation (\ref{anti-PT}), we can take a subset $S_{{\bfi},{\bfj}}$ of $[n]$ such that
$$
[W^{{\bfi},{\bfj}}]^{T(S_{{\bfi},{\bfj}})}
=X_{\bfi}(d_{\bfj}p_{\bfi},d_{\bfj}q_{\bfi},0)+X_{\bfi}(0,0,c_{\bfi}v_{\bfj})=
X_{\bfi}(d_{\bfj}p_{\bfi},d_{\bfj}q_{\bfi},c_{\bfi}v_{\bfj}).
$$
Furthermore, $[W^{{\bfi},{\bfj}}]^{T(S_{{\bfi},{\bfj}})}$ is positive since
$$
\sqrt{(d_{\bfj}p_{\bfi})(d_{\bfj}q_{\bfi})}
=d_{\bfj}\sqrt{p_{\bfi}q_{\bfi}}
=\frac{|v_{\bfj}|\sqrt{p_{\bfi}q_{\bfi}}}{\sum_{{\bfk}\in S_-}|v_{\bfk}|}
\ge\frac {|v_{\bfj}|\sqrt{p_{\bfi}q_{\bfi}}} {\sum_{{\bfk}\in S_+}\sqrt{p_{\bfk}q_{\bfk}}}
= c_{\bfi}|v_{\bfj}|.
$$
This shows that
$$
W=W_+ +\sum_{{\bfj}\in S_-}\, \sum_{{\bfi}\in S_+}W^{{\bfi},{\bfj}}
$$
is decomposable.

For the general cases, we apply the preceding argument to $W+\varepsilon I$ for $\varepsilon>0$. Then,
$W+\varepsilon I$ is decomposable if and only if $\lan\varrho,W+\varepsilon I_{2^n}\ran \ge 0$ for every fully bi-separable state $\varrho$
if and only if
$$
\sum_{{\bfi}\in B_0} \sqrt{(s_{\bfi} +\varepsilon)(t_{\bfi}+\varepsilon)} \ge \sum_{{\bfi}\in B_0} |u_{\bfi}|.
$$
Since the choice of $\varepsilon>0$ is arbitrary, the conclusion follows because the convex cone
of all decomposable matrices is closed.
\end{proof}

We note that the inequality (\ref{ineq_gew}) is much more stronger than the inequality (\ref{ineq_decom}),
and see the following.

\begin{corollary}\label{iuyuhkj}
Every {\sf X}-shaped multi-qubit genuine entanglement witness is decomposable.
\end{corollary}


\section{Discussion}

In this paper, we have characterized various kinds of separability and witnesses for {\sf X}-shaped multi-qubit matrices.
As for states, we have the following diagram for implications
between various notions of separability:

$$
\begin{array}{ccccccc}
\mbox{fully separable}
&\Longrightarrow
& \mbox{fully bi-separable}
&\Longrightarrow
& \mbox{bi-separable}
&&\\
&& \Downarrow && \Downarrow\\
&&\mbox{PPT}
&\Longrightarrow
&\mbox{PPT mixture}
&\Longrightarrow
&\mbox{state}
\end{array}
$$
Corollary \ref{esrdh} and Theorem \ref{X-PPT} tell us that the vertical arrows in the above diagram
become actually equivalences for {\sf X}-shaped multi-qubit states.
It is also known that every PPT mixture is bi-separable for some special subclasses of multi-qubit states
\cite{{guhne11},{novo}}. The authors do not know if the converses of two vertical arrows hold or not in general multi-qubit cases.
Examples of three qubit fully bi-separable states which are not fully separable can be found in \cite{kye_multi_dual,abls,brunner}.

In order to consider the dual diagram, we adopt the following terminologies:
We call a Hermitian matrix $W$
\begin{itemize}
\item
{\sl block positive} if $\lan\varrho, W\ran\ge 0$ for every fully separable state $\varrho$.
\item
{\sl bi-block positive} if $\lan\varrho, W\ran\ge 0$ for every fully bi-separable state $\varrho$.
\item
{\sl fully bi-block positive} if $\lan\varrho, W\ran\ge 0$ for every bi-separable state $\varrho$.
\item
{\sl fully bi-decomposable} if $\lan\varrho, W\ran\ge 0$ for every PPT mixture $\varrho$.
\end{itemize}
With these terminologies, a genuine entanglement witness is nothing but a non-positive fully bi-block positive matrix.
Now, we have the following dual diagram:

$$
\begin{array}{ccccccc}
\mbox{block positive}
&\Longleftarrow
& \mbox{bi-block positive}
&\Longleftarrow
& \mbox{fully bi-block positive}
&&\\
&& \Uparrow && \Uparrow\\
&&\mbox{decomposable}
&\Longleftarrow
&\mbox{fully bi-decomposable}
&\Longleftarrow
&\mbox{positive}
\end{array}
$$
The vertical arrows are again equivalences for {\sf X}-shaped multi-qubit witnesses,
by Theorem \ref{decom} and Theorem \ref{main}.
In the tri-partite case of $M_p\otimes M_q\otimes M_r$, the three kinds of positivity in the diagram
may be interpreted as various kinds of
positivity of the corresponding bi-linear maps \cite{han_kye_tri}: $W$ is fully bi-block positive if and only if it is
the Choi matrix of a bi-linear map which is $(p,q,1)$, $(1,q,r)$ and $(p,1,r)$-positive simultaneously.
Bi-block positivity  corresponds to the convex hull of these three kinds of positivity. Finally,
$W$ is block positive if and only if the corresponding bi-linear map is $(1,1,1)$-positive.

One of the merits to consider {\sf X}-shaped witnesses is to get necessary conditions for various kinds of
separability in terms of diagonal and anti-diagonal entries, as in Corollary \ref{hhhhhh}.
This necessary condition for bi-separability is also sufficient for multi-qubit {\sf X}-shaped states, as one may see in
Corollary \ref{esrdh}. We also see the same reasoning for full bi-separability in Theorem \ref{X-PPT}, because
the inequality (\ref{ineq_ppt}) is necessary for full bi-separability and PPT of general multi-qubit states.

We note that the same situations occur for witnesses. For example, we have shown that
the inequality (\ref{ineq_gew}) is equivalent to full bi-block positivity for multi-qubit {\sf X}-shaped case. One can show that
every fully bi-block positive multi-qubit Hermitian matrix satisfies this inequality. To show this,
we can expand the pairing with the following state
$$
X_{\bfi}\left(\sqrt {t_{\bfi} \over s_{\bfi}},\ \sqrt {s_{\bfi} \over t_{\bfi}},\ -e^{-i\theta_{\bfi}}\right)
+X_{\bfj}\left( \sqrt {t_{\bfi} \over s_{\bfi}},\ \sqrt {s_{\bfj} \over t_{\bfj}}\ -e^{-i\theta_{\bfj}}\right),
$$
which is bi-separable by Proposition \ref{ST_sep}.
It is also easy to see that the inequality (\ref{ineq_decom}) is necessary for decomposability of arbitrary
multi-qubit Hermitian matrices.

It would be interesting to look for necessary and sufficient conditions for full separability and block positivity
of {\sf X}-shaped multi-qubit matrices in term of entries.
See Theorem 6.3 in \cite{han_kye_tri} in this direction.
It is also natural to ask if these conditions
are necessary for corresponding properties in general cases.

\bigskip

{\sl Note added in proof}. Recently, an analytic example of three qubit PPT mixture which is not bi-separable has been constructed in \cite{ha_kye}.

\end{document}